\documentclass[12pt]{article}
\usepackage[numbers,sort]{natbib}
\usepackage{fullpage}
\usepackage{amsmath, amsthm, amssymb}
\usepackage{verbatim}
\usepackage[pagebackref=false,colorlinks=false]{hyperref}
\hypersetup{linkcolor=blue,filecolor=blue,citecolor=blue,urlcolor=blue}
\usepackage{amsmath,latexsym,bm,mathtools}
\usepackage{amsthm} % for *proof*
\usepackage{thmtools} % for restate the theorems
\usepackage{mathtools} % for coloneqq:=

\usepackage[vlined,ruled,linesnumbered]{algorithm2e} % For algorithms
\SetCommentSty{small}

\usepackage[capitalize,noabbrev]{cleveref}

\usepackage{enumitem}

\usepackage[T1]{fontenc}
\usepackage[tt=false,type1=true]{libertine}

\usepackage{tikz}  %needed for \textcolor as well as tikz
\usepackage{pgfplots}
\pgfplotsset{compat=1.18}
\usepackage{tcolorbox}

\usepackage{bbm}

\theoremstyle{plain}
\newtheorem{theorem}{Theorem}[section]

\newtheorem{lemma}[theorem]{Lemma}

\theoremstyle{definition}
\newtheorem{definition}[theorem]{Definition}

\newtheorem{example}[theorem]{Example}

\theoremstyle{remark}

\usepackage{mathtools}
\usepackage{tcolorbox}

\usepackage{subcaption}

\DeclareMathOperator*{\bbE}{{\mathbb{E}}}

\newcommand{\bbR}{\mathbb{R}}

\DeclareMathOperator*{\argmax}{argmax}

\newcommand{\bfb}{\mathbf{b}}

\newcommand{\bfx}{\mathbf{x}}

\newcommand{\LW}{\mathtt{LW}}
\newcommand{\OBJ}{\mathtt{OBJ}}
\newcommand{\UE}{\mathtt{UE}}
\newcommand{\REV}{\mathtt{REV}}

\newtheorem{assumption}{Assumption}[section]
\crefname{assumption}{assumption}{assumptions}
\Crefname{assumption}{Assumption}{Assumptions}

\DeclareMathOperator*{\E}{{\mathbb{E}}}
\DeclareMathOperator{\I}{{\mathbb{I}}}

\usepackage{wrapfig}
\usepackage{makecell}

\usepackage{multirow}

\begin{document}

\title{Beyond Advertising: Mechanism Design for Platform-Wide Marketing Service ``QuanZhanTui''}

\author{Ningyuan Li$^{1}$\thanks{This work is done during internship at Alibaba Group.}, Zhilin Zhang$^2$, Tianyan Long$^1$, Yuyao Liu$^2$,
Rongquan Bai$^2$,\\ Yurong Chen$^{3}$\footnotemark[1], Xiaotie Deng$^1$, Pengjie Wang$^2$, Chuan Yu$^2$, Jian Xu$^{2}$\thanks{Corresponding author.}, Bo Zheng$^2$\\
$^1$CFCS, School of Computer Science, Peking University\\
$^2$Alibaba Group\\
$^3$Inria, Ecole Normale Supérieure, PSL Research University
\renewcommand\footnotemark{}
\setcounter{footnote}{3}
\thanks{Authors' emails:\{liningyuan, xiaotie\}@pku.edu.cn, romeolong@stu.pku.edu.cn, \{zhangzhilin.pt, yuyao.lyy, rongquan.br, pengjie.wpj, yuchuan.yc, xiyu.xj, bozheng\}@alibaba-inc.com, 
yurong.chen@inria.fr}
}

% \email{
% {liningyuan, romeolong, xiaotie}@pku.edu.cn}
% \email{
% {zhangzhilin.pt, yuyao.lyy, rongquan.br, pengjie.wpj, yuchuan.yc, xiyu.xj, bozheng}@alibaba-inc.com}
% \email{
% yurong.chen@inria.fr
% }
\date{}
\maketitle

\begin{abstract}
On e-commerce platforms, sellers typically bid for impressions from ad traffic to promote their products. However, for most sellers, the majority of their sales come from organic traffic. Consequently, the relationship between their ad spending and total sales remains uncertain, resulting in operational inefficiency. To address this issue, e-commerce platforms have recently introduced a novel platform-wide marketing service known as QuanZhanTui, which has reportedly enhanced marketing efficiency for sellers and driven substantial revenue growth for platforms. QuanZhanTui allows sellers to bid for impressions from the platform's entire traffic to boost their total sales without compromising the platform's user experience. In this paper, we investigate the mechanism design problem that arises from QuanZhanTui. The problem is formulated as a multi-objective optimization to balance sellers' welfare and platform's user experience. We first introduce the stock-constrained value maximizer model, which reflects sellers' dual requirements on marketing efficiency and platform-wide ROI. Then, we propose the Liquid Payment Auction (LPA), an auction designed to optimize the balanced objectives while accounting for sellers' requirements in the auto-bidding environment. It employs a simple payment rule based on sellers' liquid welfare, providing a clearer link between their investment and total sales. Under mild assumptions, we theoretically prove desirable properties of LPA, such as optimality and incentive compatibility. Extensive experiments demonstrate LPA's superior performance over conventional auctions in QuanZhanTui.
\end{abstract}
\footnotetext[1]{This is the authors' version of the work. The definitive version was published in Proceedings of the 31st ACM SIGKDD Conference on Knowledge Discovery and Data Mining V.2 (KDD '25), \url{https://doi.org/10.1145/3711896.3736853}.}
\section{Introduction}
In recent years, a novel marketing service \emph{QuanZhanTui} has been implemented by leading e-commerce platforms such as Taobao, Pinduoduo, and JD.com, and has been associated with reported improvements in marketing efficiency and revenue growth~\cite{AlibabaGroup_2024Q3_Results,PDD_2024Q3_Results,JD_2024Q3_Results}.
The emergence of QuanZhanTui gives rise to a set of critical questions: What distinguishes QuanZhanTui from conventional advertising? How should the mechanism be designed for it? And how does the designed mechanism perform compared to the classical ones? This paper discusses these questions.
% This paper seeks to explore these questions and provide a theoretical framework for understanding and optimizing auction design for the QuanZhanTui scenario.
% Leading platforms such as Taobao, Pinduoduo, and JD.com have launched this service, reporting enhanced marketing efficiency for sellers and substantial revenue growth for the platforms~\cite{AlibabaGroup_2024Q3_Results,PDD_2024Q3_Results,JD_2024Q3_Results}. What exactly is QuanZhanTui, and why does it have a profound impact? What unique challenges does it pose in mechanism design, and how should a mechanism be specifically designed for QuanZhanTui? This paper discusses these questions.

Conventionally, like many other internet applications such as content platforms and social networks, e-commerce platforms reserve a certain portion of the platform-wide traffic as ad traffic for monetization. Sellers on e-commerce platforms can bid for impressions from ad traffic to promote sales. In other words, their ad spending can influence the sales from ad traffic. However, ad traffic generally contributes only 10\%--20\% of total sales~\cite{AlibabaGroup_2024Q3_Results,PDD_2024Q3_Results}. For most sellers, the majority of their sales come from organic traffic, such as search and recommendation. It introduces a critical challenge to sellers: the relationship between their ad spending and total sales is uncertain, resulting in operational inefficiency.

% While the ad system directly affects only only 15-30\%(need cite)  of total sales, empirical evidence indicates that the majority of sellers' sales are generated through organic impressions (e.g., recommendation systems)
% Under this conventional paradigm, 

% Consequently, the relationship between a seller’s ad spending and their overall platform-wide sales remains unclear, making it difficult for sellers to optimize their advertising strategies, inventory management, and operational efficiency.

% E-commerce platforms typically employ a dual monetization strategy: imposing sales commissions on sellers while constructing ad systems that allocate dedicated ad impressions for sponsored promotions. However, this conventional practice introduces a critical operational challenge for sellers. While the ad system directly affects only only 15-30\%(need cite)  of total sales, empirical evidence indicates that the majority of sellers' sales are generated through organic impressions (e.g., recommendation systems). This fundamental disconnect between advertising expenditures and platform-wide total sales results in significant operational inefficiencies for sellers' business.

% \begin{proof}

\begin{example}[{\sffamily{CONVENTIONAL ADVERTISING}}]\label{ex-1}
An apparel seller has stocked 100,000 down jackets in preparation for the winter demand surge. The seller anticipates selling 90,000 units through organic traffic and the remaining 10,000 units through ad traffic. However, organic traffic might only yield 50,000 units being sold. In this case, even if the seller increases ad spending so that 20,000 units can be sold through ad traffic, there are still 30,000 unsold units. To cover storage and labor costs (e.g., up to 50\% of the product price), the seller has to liquidate the inventory at a loss. Figure \ref{fig:fig2}(a)(b) gives a conceptual illustration of this example. $\hfill\square$

%Even with an unlimited advertising budget, only 20,000 units might be sold through ad traffic during the entire winter season. While the seller anticipates selling 80,000 units through organic traffic, unpredictable organic traffic might result in only 30,000 units being sold, leaving 50,000 units as unsold inventory. The out-of-season inventory incurs significant costs (up to 50\% of the price per product), including storage and labor expenses, which forces the seller to liquidate remaining inventory at a loss. This results in substantial financial losses that ultimately erode initial profits. This scenario highlights the inefficiency of the traditional marketing paradigm, which exposes sellers to significant financial risks, such as overstocking and margin erosion. $\hfill\square$
\end{example}

This persistent challenge continues to vex sellers as they are unable to directly drive sales from organic traffic. To address it, QuanZhanTui has been introduced as a platform-wide marketing service that allows sellers to bid for impressions from the platform's entire traffic without compromising the platform's user experience.

\begin{example}[{\sffamily{QUANZHANTUI}}]
% \noindent\textbf{EXAMPLE 2 (QuanZhanTui)}: 
With QuanZhanTui, the seller in \cref{ex-1} can directly influence her total sales on the platform. She can enhance the down jacket's competitiveness in impression acquisition throughout the entire platform by simply increasing marketing spending\footnote{We deliberately use a different term to distinguish it from ad spending in conventional advertising scenarios.}. In other words, she can have better control over platform-wide sales, so that the 100,000 units can be sold as long as enough marketing spending is invested. We note that even if the marketing spending is more than the ad spending in \cref{ex-1}, the elimination of unsold inventory guarantees a better profit. Figure \ref{fig:fig2}(d)(e) gives a conceptual illustration of this example.$\hfill\square$

%Suppose the seller in Example 1 can now influence the platform-wide sales by increasing marketing spending, they can enhance the product's competitiveness directly across overall platform. In other words, seller can gain explicit control over platform-wide sales, enabling them to successfully sell all 100,000 inventory. Although the marketing spending may be higher than before, the elimination of inventory storage costs renders more profitable while ensuring greater certainty.$\hfill\square$
\end{example}

The rationality of QuanZhanTui is based on two characteristics of e-commerce platforms. First, ads on these platforms inherently consist of products available on the platforms and are often highly targeted based on user data. Therefore, unlike those on content platforms and social networks, they exhibit contextually native integration and only cause negligible user experience interruptions. Second, sellers on these platforms seek a more certain relationship between their investment and platform-wide sales. Their profitability depends on the total sales rather than those from either the ad traffic or the organic traffic.

Contemporary internet applications, including e-commerce platforms, have been trying to optimize multiple objectives such as revenue and user experience by mixing organic items and ads in search or recommendation feeds \cite{DBLP:conf/recsys/LinCPSXSZOJ19,DBLP:conf/kdd/YanXTC20,DBLP:conf/aaai/CarrionWNLLGLCJ23,DBLP:conf/www/LiMZW0Y00D24}. QuanZhanTui is fundamentally different from this practice. The difference can be illustrated by Figure \ref{fig:fig2}(c)(f). Mixing organic items and ads still requires two separate systems to retrieve the corresponding items (organic or sponsored) based on different criteria. Then the retrieved organic items and ads are mixed in pursuit of some Pareto optimization. However, sellers still lack direct means to influence their sales from organic traffic. QuanZhanTui no longer distinguishes between organic traffic and ad traffic. Platform-wide sales of an item can be directly influenced, as long as the seller chooses to sponsor it.

QuanZhanTui essentially grants sellers operational autonomy to better trade off their marketing spending and platform-wide sales based on their profit margins and marketing strategies. Since it makes the relationship between marketing spending and platform-wide sales more predictable, it also motivates sellers to engage more actively with the platform. Consequently, it has quickly become the preferred marketing service for sellers, and meanwhile driven significant revenue growth for e-commerce platforms.

\begin{figure*}[t]
\centering
\includegraphics[width=1.0\textwidth]{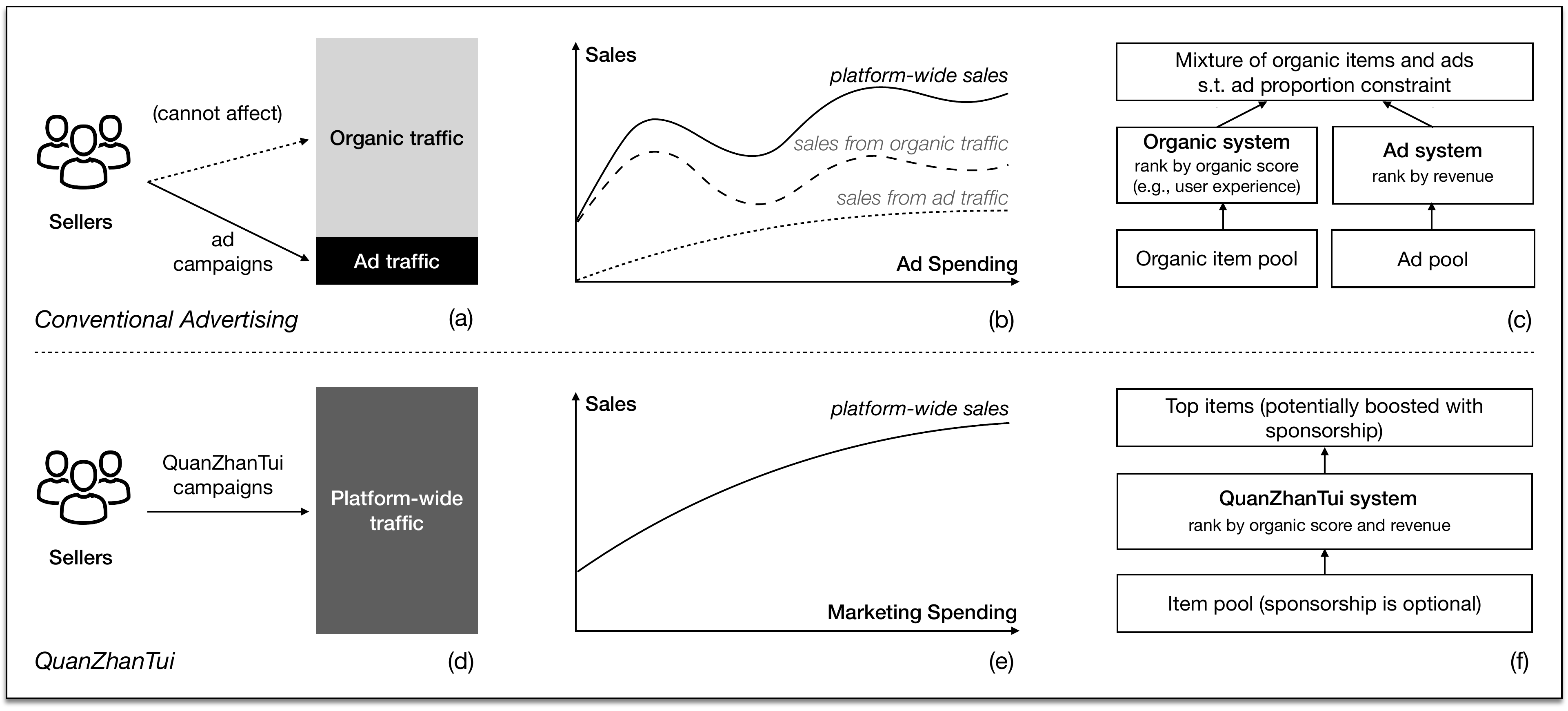}
% \vspace{-.5\baselineskip}
\caption{The difference between Conventional Advertising and QuanZhanTui.}
%\vspace{-0.5\baselineskip}
\label{fig:fig2}
\end{figure*}

% \begin{figure*}[t]
%     \centering
%     \includegraphics[width=.75\linewidth]{Figs/framework.pdf}
%     \caption{Overall Framework for Generative Auto-bidding.} 
    
%     % We construct the relationship between condition and trajectory through the diffusion process. When generating the bidding parameters, we simply use $p_\theta$ with expected conditions and history states to predict future states and generate the bidding parameters through inverse dynamics $f_\phi$.}
%     \label{fig:model}
% \end{figure*}

The success of QuanZhanTui is heavily dependent on the design of the mechanisms. However, no systematic study has investigated the mechanism design problem that arises from QuanZhanTui and proposed a solution. This paper is the first effort. The contributions of this paper can be summarized as follows.
\begin{itemize}
    \item We introduce QuanZhanTui, an emerging platform-wide marketing service in e-commerce platforms, to the research community and formulate the mechanism design problem associated with it.
    \item We propose the stock-constrained value maximizer model and the Liquid Payment Auction (LPA). Under mild assumptions, LPA theoretically leads to the optimal platform objectives and incentives sellers to be truthful in QuanZhanTui. 
    \item We conducted extensive experimental evaluations to demonstrate that LPA outperforms classical auctions in the QuanZhanTui scenario. LPA is also simple to implement, which is crucial for its industrial-scale application.
\end{itemize}

\subsection{Related Work}

\textbf{Auto-bidding under advertiser constraints} has been a pivotal research direction in computational advertising \cite{DBLP:journals/ior/BalseiroKMM21,pmlr-v119-balseiro20a,ijcai2023p314liu,10.1145/3543507.358349feng,DBLP:conf/icml/CastiglioniCK24,DBLP:conf/kdd/GuoHZWYXZZ24}. 
Our work builds on the concepts of uniform bidding and pacing equilibrium.
\cite{DBLP:conf/wine/AggarwalBM19} showed that uniform bidding is optimal for truthful auctions, and \cite{DBLP:journals/mansci/BalseiroG19} studied adaptive pacing strategy in repeated auctions. \cite{DBLP:journals/ior/ConitzerKSM22,DBLP:conf/sigecom/ChenKK21,chen2025constantinapproximabilitypacingequilibria} study pacing equilibrium in second-price auctions, showing its computational intractability. In contrast, \cite{ConitizerFPApe} showed that the pacing equilibrium in first-price auctions is efficiently computable and offers desirable properties. 

\textbf{Mechanism design under advertiser constraints} has been extensively explored in the previous literature, primarily aiming to maximize revenue \cite{DBLP:conf/sigecom/BalseiroBW16,DBLP:conf/sigecom/BalseiroDMMZ21,DBLP:conf/wine/LvBZW23} or liquid welfare \cite{DBLP:conf/icalp/DobzinskiL14,DBLP:conf/sagt/LuX17}. Some of these works addressed the incentive compatibility of private constraints, while others assumed public constraints. \cite{ZhaohuaChenWWW24} examined a budget-constrained auction design setting involving strategically manipulable priors.
In auto-bidding, the literature analyzed the efficiency guarantees of learning processes \cite{gaitonde_et_al:LIPIcs.ITCS.2023.52,pmlr-v247-lucier24a,DBLP:conf/sigecom/FikiorisT23}, and studied the efficiency of auctions in auto-bidding equilibrium \cite{DBLP:conf/www/DengMMZ21,DBLP:conf/www/Mehta22, DBLP:conf/www/LiawMZ24,DBLP:conf/www/DengMMMZZ24}. Critically, incentive issues may arise from potential misreporting of constraints in auto-bidding systems when constraints are private to advertisers \cite{DBLP:conf/ijcai/0012T24, DBLP:conf/sigecom/AlimohammadiMP23}. To our best knowledge, this paper is the first to propose a novel auction specifically designed for auto-bidding environments that incentivizes advertisers to truthfully report their constraints.

\textbf{Automated mechanism design for conventional advertising}
has been widely employed in contemporary ad systems \cite{DBLP:conf/wsdm/ZhangLZZXPYWXG21,DBLP:conf/kdd/LiuYZZRLHWCXWCZ21,DBLP:conf/kdd/LiMZ00ZXZD23,DBLP:conf/cikm/ShenSGSYWSN23}, typically aiming to optimize the trade-off among multiple objectives such as revenue and user experience \cite{DBLP:conf/sigecom/BachrachCKKK14}. 
In addition, the conventional separation of ads and organic items motivated a series of research on the design of mechanisms to mix them together so that some performance index is maximized under certain constraints \cite{DBLP:journals/tcs/LiQWY23,DBLP:conf/recsys/LinCPSXSZOJ19,DBLP:conf/kdd/YanXTC20,DBLP:conf/aaai/CarrionWNLLGLCJ23,DBLP:conf/www/LiMZW0Y00D24}. QuanZhanTui is fundamentally different as it no longer distinguishes between organic traffic and ad traffic.

\section{Problem Formulation}

% \begin{figure*}[t]
%     \centering
%     \includegraphics[width=\textwidth]{figures/quanzhantui_example_2.pdf}
%     \caption{
%         slots.
%     }
%     \label{fig:ABPlanner}
% \end{figure*}

\label{sec:model}

\textbf{Auction process.} We model the auction process for all traffic from the platform as a repeated multi-slot auction setting. 
There are $n$ sellers indexed by $i\in[n]$, and $T$ user requests $q_1,\cdots,q_T$. Each request corresponds to an auction allocating $m$ slots indexed by $j\in[m]$. 
Without loss of generality, we assume that each seller sells only one product. When a seller $i$'s product is allocated with the $j$-th slot in request $q_t$, the expected sales, typically measured by gross merchandise volume (GMV), are represented as $\beta_jv_{i}(q_t)$, where $\beta_j$ denotes the exposure rate of the $j$-th slot, and $v_{i}(q_t)$ denotes the GMV after exposure. The product is also associated with a user experience metric $e_i(q_t)$ after exposure, which can capture various factors such as click-through rate (CTR), GMV, or their combination, reflecting how well the product aligns with user needs. 
% represent whether the displayed results satisfy user needs, such as click-through rate (CTR), or whether they meet purchasing requirements, such as gross merchandise volume (GMV), or a weighted combination of both. 
In conventional advertising, these slots are typically divided into organic traffic and ad traffic, with the ad auction only processing the ad traffic. 
In contrast, in QuanZhanTui, the auction operates across all traffic, controlling the allocation of all slots rather than treating organic items and ads separately. \looseness=-1
%
% In contrast, in QuanZhanTui, the auction operates across all traffic, controlling the allocation of all slots. 
% \lny{$e_i(q_t)$ is ... ,also affect the auction. }

The platform runs an auction for each request $q_t$ to allocate each slot to exactly one seller's product. Sellers compete in the auctions through the auto-bidding service provided by the platform. For each seller $i$, the platform deploys an auto-bidder, also indexed by $i$, to bid in auctions on behalf of the seller. More specifically, for each request $q_t$, each auto-bidder $i$ submits a bid $b_{i,t}\geq 0$, constituting a bid profile $\bfb_t=(b_{1,t},\cdots,b_{n,t})$. The auction determines the allocation $(x_{i,j}(\bfb_t;q_t))_{i\in[n],j\in[m]}$ and the payments $(p_{i}(\bfb_t;q_t))_{i\in[n]}$. Here $x_{i,j}(\bfb_t;q_t)\in\{0,1\}$ indicates whether the $j$-th slot is allocated to the seller $i$'s product, and $p_{i}(\bfb_t)$ is the payment charged from seller $i$ if some slot is allocated to her. 
It is worth noting that the allocation and payment rules may depend on $v_i(q_t)$ and $e_i(q_t)$ for all $i\in[n]$, which are known to the platform once $q_t$ is given.
% It is worth noting that $e_i$ is also one of the important factors in determining the allocation and payment in the auction. However, for simplicity, we omit $e_i$ in the expression for $x_i$. 
In addition, sellers who do not choose to use QuanZhanTui can still participate in the auction, which means $b_{i,t} = 0$. They still have the opportunity to win the auction if their $e_i(q_t)$ is large enough.
For convenience, we define a seller's allocated exposure rate as
$$x_i(\bfb_t;q_t):=\sum_{j\in[m]}\beta_j x_{i,j}(\bfb_t;q_t).$$
%\looseness=-1

\textbf{Seller model.}  
To model the marketing paradigm of sellers in the QuanZhanTui scenario, we propose a novel seller model called \emph{stock-constrained value maximizer}.
Consider a seller who sells a product with finite stock, represented by a maximum potential sales value $S_i$, calculated as the stock times the unit price of the product. 
% The seller has a target ROI ratio of $\tau_i>0$. 
% Consider a seller who sells a product with finite stock $S_i>0$ and a target ROI ratio $\tau_i>0$. 
% The seller aims to maximize sales while ensuring that the realized platform-wide ROI (rather than the advertising ROI) does not exceed the ROI constraint.
The seller aims to maximize sales, while ensuring that the realized platform-wide ROI
is at least $\tau_i$ for some target ROI ratio $\tau_i>0$. 
Importantly, sales cannot exceed the available stock. Therefore, the seller is assumed to be indifferent to any total GMV above $S_i$, which means that sales beyond this point do not contribute additional value.
%\looseness=-1

Formally, seller $i$ and the corresponding auto-bidder $i$'s target can be formulated as the following optimization problem:
% We consider value-maximizing sellers with budget and ROI constraints throughout the paper.
% Each seller $i\in[n]$ has a budget constraint $B_i>0$ and an ROI constraint $\tau_i>0$, aiming to maximize her total gmv under the constraints. Formally, seller $i$ and the corresponding auto-bidder $i$'s objective can be formulated as the following optimization problem:
\begin{align*}
&\max_{(b_t)_{t\in[T]}}\min\left\{S_i,\sum_{t\in[T]}x_i(\bfb_t;q_t) v_{i}(q_t)\right\}\\
&\text{s.t. } \sum_{t\in[T]}p_{i}(\bfb_t;q_t)\leq \frac1{\tau_i}\min\{S_i,\sum_{t\in[T]}x_i(\bfb_t;q_t) v_{i}(q_t)\}.
\end{align*}

Notably, the stock constraint can be viewed as a variation of the more commonly studied budget constraint. We can observe from the constraint that $\sum_{t\in[T]}p_{i}(\bfb_t;q_t)\leq \frac{S_i}{\tau_i}$, which can be interpreted as a budget constraint $B_i=\frac{S_i}{\tau_i}$. 
Since e-commerce platforms typically require sellers to set constraints in terms of budget and ROI, we will describe the stock-constrained seller model using $B_i$ and $\tau_i$ throughout the paper, assuming that $S_i=\tau_iB_i$.
Under this assumption, the stock-constrained seller's target is equivalently expressed as 
\begin{align*}
&\max_{(b_t)_{t\in[T]}}  \min\left\{B_i,\frac1{\tau_i}\sum_{t\in[T]}x_i(\bfb_t;q_t) v_{i}(q_t)\right\}\\
&\text{s.t. } \sum_{t\in[T]}p_{i}(\bfb_t;q_t)\leq \min\{B_i,\frac1{\tau_i}\sum_{t\in[T]}x_i(\bfb_t;q_t) v_{i}(q_t)\}
\end{align*}

% In comparison, the traditional value maximizing seller model with budget and ROI constraints \cite{} can be formulated as
% \begin{align*}
% &\max_{(b_t)_{t\in[T]}}\sum_{t\in[T]}x_i(\bfb_t;q_t) v_{i}(q_t)\\
% \text{s.t. }&\sum_{t\in[T]}p_{i}(\bfb_t;q_t)\leq \min\{B_i,\frac1{\tau_i}\sum_{t\in[T]}x_i(\bfb_t;q_t) v_{i}(q_t)\}
% \end{align*}
% which only differs in that the value will not saturate.

% This also aligns

Based on the seller model, the auto-bidder dynamically adjusts the bidding strategy during the repeated auctions, to maximize seller $i$'s objective function while ensuring that the budget and ROI constraints are satisfied.

\textbf{Information assumptions.}
We assume that each request $q_t$ is independently drawn from a static distribution $D$ over a space $Q$. 
The distribution $D$ is unknown to the platform. However, given any request $q_t$, we assume that the platform can compute the expected sales $v_i(q_t)$ and user experience $e_i(q_t)$ precisely, without any prediction error. 

We allow the allocation rules $x_{i,j}(\bm{b}_t;q_t)$ and payment rules and $p_i(\bm{b}_t;q_t)$ to be dependent on the budget and ROI constraints set by sellers, as these information are submitted to the platform in the QuanZhanTui scenario.
Unless otherwise specified, we assume that the sellers submit their budget and ROI constraints truthfully. This assumption will be justified by the incentive compatibility result stated in \Cref{thm:IC-for-stock-constrained}.

\textbf{Platform objective.}
In the QuanZhanTui scenario,  the platform seeks to achieve fair and efficient allocation and payment outcomes throughout the auction process. For the prosperity of the platform, we define the mechanism design's objective as a balanced optimization of the sellers' welfare and the user experience. Formally, the objective is expressed as a linear combination of total liquid welfare and user experience metrics, thereby incorporating both economic and organic factors:
\begin{align}
    \OBJ=\LW+\kappa\cdot \UE, 
\end{align}
where $\kappa>0$ is a constant coefficient. 

The liquid welfare is a fair metric of overall welfare widely adopted for budget and ROI constrained sellers, defined as $$\LW=\sum_{i\in[n]}\min\left\{B_i,\frac1{\tau_i}\sum_{t\in[T]}x_i(\bfb_t;q_t) v_{i}(q_t)\right\}.$$
It serves as a reasonable welfare metric for stock-constrained sellers.

The total user experience metric is defined as
$$\UE=\sum_{t\in[T]}\sum_{i\in[n]}x_i(\bfb_t;q_t) e_{i}(q_t).$$
% The linear form is without loss of generality as it is able to recover the Pareto frontier \cite{}.
%
% Alternatively, we may also take the platform revenue into consideration. The objective can be defined as
% \begin{align}
%     \OBJ^R=(1-\kappa)\LW+\kappa\REV+\kappa\cdot \UE,
% \end{align}
% where $\kappa\in(0,1)$ is another constant coefficient, and the revenue is the total payment from sellers: $$\REV=\sum_{t\in[T]}\sum_{i\in[n]}\sum_{j\in[m]}p_{i}(\bfb_t;q_t).$$
Here we also define the platform revenue, which is the total payment from sellers: $$\REV=\sum_{t\in[T]}\sum_{i\in[n]}p_{i}(\bfb_t;q_t).$$

% In most parts of this paper, similar to many literature involving auto-bidding, we assume that the budget constraints and ROI constraints are also known to the platform. 
% In section 5, we will discuss the cases where the budget and/or ROI constraints are private to each seller and unknown to the platform, where the issue of incentive compatibility arises. That is, the sellers may misreport their private constraints to the platform if this improves their objective. 
\textbf{Auto-bidding environment with uniform bidding.}
Generally speaking, the mechanism for QuanZhanTui is a joint system consisting of the auction and the auto-bidding algorithms.
However, to simplify the problem, we restrict the auto-bidding system to use uniform bidding strategies and expect the auction to perform well under the equilibrium state of the auto-bidding system. Our theoretical results will show that this assumption does not compromise the outcome. \looseness=-1
%\looseness=-1

Formally, a uniform bidding strategy $\bfb$ is expressed in the form of $\smash{b_{i,t}=a_i^{(t)}\frac1{\tau_i}v_i(q_t)}$, where the multiplicative factor $\smash{a_i^{(t)}}$ is called the pacing factor. Typically, $\smash{a_i^{(t)}}$ is updated by some pacing algorithm after each auction round, according to the auction outcome and the seller $i$'s target and constraints. The definition of equilibrium will be presented in \Cref{sec:analysis}.

\textbf{Auction design problem.}
Our goal is to design an auction that optimizes the platform objective $\OBJ$ throughout the platform-wide traffic allocation process.
It should be compatible with the auto-bidding environment adopting uniform bidding, and perform optimally under the equilibrium state.
It is also important for the auction to ensure the joint system's incentive compatibility (IC) for sellers, preventing strategic manipulation when setting budget and ROI constraints.
% The proposed auction must have a clear and direct connection between marketing spending and platform-wide sales during the allocation process. 
Additionally, to provide a clearer relationship between marketing spending and platform-wide sales, a simple and stable payment rule is desirable.
% To ensure that seller constraints (e.g., platform-wide ROI) are satisfied and that the auction remains resistant to strategic manipulation, the auction needs to have a stable payment rule. 

% We aim to design an auction, such that:\lny{Need more notation to formally write aims?}
% \begin{enumerate}
%     \item Mainstream auto-bidding algorithms converge to a predictable equilibrium.
%     \item The platform objective under the equilibrium is optimized.
% \end{enumerate}

\section{Liquid Payment Auction}
In this section, we propose the \emph{Liquid Payment Auction} (LPA), which is a simple yet optimal auction designed to maximize the platform's objective in the QuanZhanTui scenario. We also provide an intuitive explanation of the motivation behind the auction rules of LPA. In the next section, we will establish the optimality of LPA through theoretical analysis. \looseness=-1
%\looseness=-1
% We conduct theoretical analysis for LPA under a Bayesian problem formulation, demonstrating that LPA achieves optimal objective for $\OBJ$ under uniform pacing equilibrium of auto-bidders.
\subsection{Auction Rule of LPA}
\begin{definition}[LPA]
Given the value functions $v_1(\cdot),\cdots,v_n(\cdot)$ and the ROI constraints $\tau_1,\cdots,\tau_n$, the Liquid Payment Auction (LPA) for each request $q_t$ is defined as follows:
\begin{enumerate}
    \item For each product $i\in[n]$, calculate a rank score $s_{i,t}=b_{i,t}+\kappa\cdot e_{i}(q_t)$.
    \item Sort the $n$ products by the rank scores $s_{1,t},\cdots,s_{n,t}$, allocate the $m$ slots to the products in descending order of rank score.
    \item For each product $i$, if it is allocated with the $j$-th slot, charge $\frac1{\tau_i}v_i(q_t)\beta_j$ from seller $i$.%\lny{$x_i()v_i(q_t)$}
\end{enumerate}
% \looseness=-1
Formally, let $\mathrm{rank}_{i}(\bfb_t;q_t)=|\{i'\in[n]\setminus\{i\}:s_{i',t}>s_{i,t}\}|+1$ denote the rank of $i$'s product in request $q_t$, the allocation and payment rules can be written as
\begin{align*}
    &x_{i,j}(\bfb_t;q_t)=\I[\mathrm{rank}_{i}(\bfb_t;q_t)=j],\\
    &p_i(\bfb_t;q_t)=\sum_{j\in[m]}\frac1{\tau_i}v_i(q_t)\beta_jx_{i,j}(\bfb_t;q_t)=\frac1{\tau_i}v_i(q_t)x_i(\bfb_t;q_t).
\end{align*}
% \[
% x_{i,j}(\bfb_t;q_t)=\I[\mathrm{rank}_{i}(\bfb_t;q_t)=j],
% \]
% and 
% \[
% p_i(\bfb_t;q_t)=\sum_{j\in[m]}\frac1{\tau_i}v_i(q_t)\beta_jx_{i,j}(\bfb_t;q_t)=\frac1{\tau_i}v_i(q_t)x_i(\bfb_t;q_t).
% \]
\end{definition}

One can see that the payment rule of LPA directly guarantees the satisfaction of sellers' ROI constraints.
To ensure the satisfaction of sellers' budget constraints, if a seller's budget is exceeded despite the adjustments by the auto-bidder, we truncate its total payment to the budget. This is most likely to occur for products with strong organic user preference, where the sales can reach the seller's stock even with minimal or no marketing investment. Specifically, the final total payment from each seller $i$ will be $\min\{B_i,\frac1{\tau_i}\sum_{t\in[T]}x_{i}(\bfb_t;q_t)v_i(q_t)\}$, which is exactly equal to seller $i$'s contribution in the liquid welfare. This payment rule is feasible in practice, as sellers submit their bidding constraints to the platform when participating in the QuanZhanTui service.

We remark that LPA shares some structural similarities with the bid-discount first price auction (BDFPA) introduced by \cite{ZhaohuaChenWWW24}, 
%examine an auction design setting involving budget-constrained agents with strategically manipulable priors. They introduce the bid-discount first price auction, which shares structural similarities with the LPA 
in which allocation ranks are determined by linearly scaled bids, while payments are based on the original bids. The key distinction lies in the source of the scaling: in BDFPA, the scaling factors are chosen by the auction designer, whereas in LPA, they emerge endogenously from auto-bidders applying uniform bidding strategies.

% We also remark that LPA is feasible in practice, as the sellers submit their constraints to the platform when accessing the QuanZhanTui service.

\subsection{Theoretical Motivation}
The main characteristic of LPA is its payment rule, which is directly determined by the seller's liquid welfare, rather than the bids submitted by the auto-bidders. 
To provide an intuitive understanding of this payment rule and the design of LPA, we describe its motivation from the perspective of maximizing the platform objective in traffic allocation.\looseness=-1

Suppose the platform directly controls the allocation of all traffic to maximize the platform objective, temporarily ignoring the bids and payments. We view $\smash{x_i^{(t)}=x_{i}(\bfb_t;q_t)}$ as the decision variable, and the feasible region of the vector $\smash{\bfx^{(t)} = (x_1^{(t)},\cdots,x_n^{(t)})}$ forms a convex set denoted by $\smash{X^{(t)}}$. Then we can write down the optimization problem of platform objective as
\begin{align*}
&\max_{(\bfx^{(t)})}\sum_{i\in[n]}\left(\min\left\{B_i, \frac1{\tau_i}\sum_{t\in[T]}v_i(q_t)x_{i}^{(t)}\right\}+\kappa\sum_{t\in[T]}e_i(q_t)x_{i}^{(t)}\right)\\
&\text{s.t. }\forall t\in[T],~ \bfx^{(t)}\in X^{(t)}
\end{align*}

% Informally, u
Using the Lagrange multiplier method, we can reformulate it as
\begin{align*}
&\max_{(\bfx^{(t)})}\sum_{i\in[n]}\left( \left( \lambda_i\frac1{\tau_i}
\sum_{t\in[T]}v_i(q_t)x_{i}^{(t)}+(1-\lambda_i)B_i \right)+
\kappa\sum_{t\in[T]}e_i(q_t)x_{i}^{(t)} \right)\\
&\text{s.t. }\forall t\in[T],~ \bfx^{(t)}\in X^{(t)}
\end{align*}
where each $\lambda_i\in[0,1]$ is a dual variable.

Under the optimal dual variables $\lambda_i^*$, the optimal primal variables are given by the following maximization problem, 
\[
\bfx^{(t)*}={\argmax}_{\bfx^{(t)}\in X^{(t)}}\sum_{i\in[n]}\left(\lambda_i^*\frac1{\tau_i}v_i(q_t)+\kappa e_i(q_t)\right)x_{i}^{(t)}.
\]

Notably, if we assume $b_{i,t}=\lambda_i^*\frac1{\tau_i}v_i(q_t)$, then the optimal allocation $\smash{\bfx^{(t)*}}$ can be obtained by allocating the slots in the descending order of the products' rank scores $s_{i,t}=b_{i,t}+\kappa e_i(q_t)$. This motivates the allocation rule of LPA.

Furthermore, if we assume that each auto-bidder $i$ adopts a uniform bidding strategy using $\lambda_i^*$ as a multiplicative pacing factor applied on $\smash{\frac1{\tau_i}v_i(q_t)}$, and let the payment from seller $i$ be $\smash{\sum_{t\in[T]}\frac1{\tau_i}v_i(q_t)}$, then by KTT conditions, $\lambda_i^*$ will be exactly the optimal pacing factor satisfying seller $i$'s budget constraint. This motivates the payment rule of LPA and the assumption of uniform bidding.
%\looseness=-1

A formal version of these arguments is presented in \Cref{thm:platform-objective-optimal}.

\section{Theoretical Analysis of LPA}
\label{sec:analysis}
%optimal platform objective; equilibrium existence
In this section, we theoretically demonstrate the advantageous properties of LPA in an auto-bidding environment adopting a uniform bidding strategy. For convenience, we introduce a Bayesian problem formulation. Instead of $T$ sequential requests, we consider the expected outcome when $q\sim D$, and relax each seller $i$'s constraints to the expectation form. A stock-constrained seller $i$'s target is then formulated as
\begin{align*}
    &\max_{b_{i}:Q\to\mathbb{R}}  \min\left\{\frac{B_i}{T}, \frac1{\tau_i}\E_{q\sim D}\left[x_{i}\left(b_1(q),\cdots,b_n(q);q\right) v_{i}(q)\right]\right\}\\
    &\text{s.t. }  \E_{q\sim D}\left[p_{i}\left(b_1(q),\cdots,b_n(q);q\right)\right]\leq \\
    &\quad\quad\min\left\{\frac{B_i}{T}, \frac1{\tau_i}\E_{q\sim D}\left[x_{i}\left(b_1(q),\cdots,b_n(q);q\right) v_{i}(q)\right]\right\}.
\end{align*}

\subsection{Pacing Equilibrium}
We assume that every auto-bidder $i$ adopts a uniform bidding strategy, that is, there exists $a_i\in[0,1]$ such that $b_i(q)=a_i\cdot\frac1{\tau_i}v_i(q)$ for all $q\in Q$. 
Under this restriction, the bidding strategy of each auto-bidder $i$ can be represented by the pacing factor $a_i$. 
We call $\bm{a}=(a_1,\cdots,a_n)\in[0,1]^n$ a pacing profile.
For convenience, we use $a_{-i}$ to denote the pacing profile of all sellers except seller $i\in[n]$.
We also slightly abuse the notation to define $$x_{i}\left(\bm{a};q\right)=x_{i}\left(a_1\cdot \frac1{\tau_1}v_1(q),\cdots,a_n\cdot \frac1{\tau_n}v_n(q);q\right),$$ and $$p_{i}\left(\bm{a};q\right)=p_{i}\left(a_1\cdot \frac1{\tau_1}v_1(q),\cdots,a_n\cdot \frac1{\tau_n}v_n(q);q\right).$$

Intuitively, since the ROI constraint is always satisfied in LPA, the adjustment of pacing factor $a_i$ aims to obtain the highest sales for seller $i$ within the budget constraint $B_i$. Raising $a_i$ will cause the payment and the obtained sales to increase simultaneously. Therefore, if a seller's total payment is below her budget $B_i$, she will increase $a_i$ until the payment reaches $B_i$ or $a_i$ hits the upper bound of $1$. Conversely, if a seller's total payment exceeds her budget, she will decrease $a_i$ until the payment reaches $B_i$ or $a_i$ hits the lower bound $0$.
A pacing equilibrium is a pacing profile under which every buyer's pacing factor is already optimal, i.e., no seller has incentive to adjust the pacing factor $a_i$. Formally, a pacing equilibrium is defined as follows. \looseness=-1
%\looseness=-1
\begin{definition}[pacing equilibrium]
\label{def:pacing equilibrium}
A pacing profile $\bm{a}\in[0,1]^n$ is called a pacing equilibrium, if for each $i\in[n]$, $a_i$ satisfies that:
\begin{enumerate}
\item If $a_i>0$, then $\E_{q\sim D}\left[p_{i}\left(\bm{a};q\right)\right]\leq \frac{B_i}{T}$.
\item If $a_i<1$, then $\E_{q\sim D}\left[p_{i}\left(\bm{a};q\right)\right]\geq \frac{B_i}{T}$.
\end{enumerate}
% Here $b_i(q)=a_i\cdot\frac1{\tau_i}v_i(q)$ for all $i\in[n]$.
\end{definition}

% \begin{remark}
% The condition for a pacing equilibrium is equivalent to that each $a_i$ is the optimal uniform bidding strategy for seller $i$.
% \end{remark}

% for convenience, we can assume that $\tau_i=1$ for all $i\in[n]$ by viewing $\frac1{\tau_i}v_i(q)$ as the new $v_i(q)$.

For simplicity, we make the following assumption, which guarantees that there is almost surely no ties, so that we can omit the discussion about tie-breaking issues.
\begin{assumption}
\label{assumption:nomass}
    For any $i_1\neq i_2\in[n]$, for any $a_{i_1},a_{i_2},\kappa\in\mathbb{R}$, it is assumed that
    \[\Pr_{q\sim D}[a_{i_1}v_{i_1}(q)+\kappa e_{i_1}(q)=a_{i_2}v_{i_2}(q)+\kappa e_{i_2}(q)]=0.\]
\end{assumption}
Note that this is a weaker assumption than the common assumption that the joint distribution of $(v_1(q),\cdots,v_n(q),e_1(q),\cdots,e_n(q))$ is a continuous distribution in $\mathbb{R}^{2n}$.

Under this assumption, we prove the following basic properties about pacing equilibrium under LPA:
\begin{theorem}
\label{thm:pacing equilibrium exist-unique}
% Given any budgets $B_1,\cdots,B_n$, ROIs $\tau_1,\cdots,\tau_n$, request distribution $D$, functions $v_1(\cdot),\cdots,v_n(\cdot)$ and $e_1(\cdot),\cdots,e_n(\cdot)$, 
Under \Cref{assumption:nomass}, the following properties about the pacing equilibrium under LPA hold:
    \begin{enumerate}
        \item At least one pacing equilibrium exists. 
        \item Moreover, there exists a unique pacing equilibrium $\bm{a}^*\in[0,1]^n$ that is point-wise maximum among all pacing equilibria.
        \item All pacing equilibria induce identical allocation and payment outcomes for every request.
        % All pacing equilibria induce the same allocations and payments in all requests. 
    \end{enumerate}
\end{theorem}

\begin{proof}
\textbf{Property 1:}
We show the existence of the pacing equilibrium by Tarski's fixed point theorem.

    For each $i\in[n]$, we define function
    \begin{align*}
        f_i(a_{-i})=\max(\{a_i\in(0,1]:\bbE_{q\sim D}[p_i(a_i,a_{-i};q)]\leq B_i/T\}\cup\{0\}).
    \end{align*} Note that by \Cref{assumption:nomass}, $\bbE_{q\sim D}[p_i(a_i,a_{-i};q)]$ is continuous in $a_i$, so $f_i(a_{-i})$ is well-defined for any $a_{-i}\in[0,1]^{n-1}$.

    Observe that $f_i(a_{-i})$ is weakly increasing in any component of $a_{-i}$, because for any $a_i$ and $q$, $p_i(a_i,a_{-i};q)$ is weakly decreasing in any component of $a_{-i}$.
    
    Then we define the mapping $\vec{f}(\bm{a})=(f_i(a_{-i}))_{i\in[n]}$, which is a monotone mapping from $[0,1]^n$ to $[0,1]^n$. By Tarski's fixed point theorem, $\vec{f}$ has at least one fixed point $\bm{a}$, i.e. $\bm{a}=\vec{f}(\vec{a})$.

    To prove the existence of pacing equilibrium, it remains to show that if $\vec{a}=f(\vec{a})$, then $\vec{a}$ is a pacing equilibrium. We check the conditions in \Cref{def:pacing equilibrium}.
    For each $i\in[n]$, if $a_i>0$, then we have $\bbE_{q\sim D}[p_i(a_i,a_{-i};q)]\leq B_i/T$ by the definition of $f_i$; if $a_i<1$, then by definition of $f_i$, it holds for all $a_i'\in(a_i,1]$ that $\bbE_{q\sim D}[p_i(a_i',a_{-i};q)]< B_i/T$, so $\bbE_{q\sim D}[p_i(a_i,a_{-i};q)]\leq B_i/T$ by continuity. Therefore, $\vec{a}$ is a pacing equilibrium.

\textbf{Property 2:}
    Now we prove that there exists a unique pacing equilibrium $\bm{a}^*$ that is point-wise maximal. To prove this, we define a set of pacing profiles $\Omega=\{\vec{a}\in[0,1]^n:\forall i\in[n],~f_i(a_{-i})\geq a_i\}$. Roughly speaking, $\Omega$ is the set of all pacing profiles under which no bidder has incentive to decrease the pacing factor. It's worth noting that the set of all pacing equilibria is a subset of $\Omega$ by definition.
    
    We first show that $\Omega$ is closed under the point-wise maximum operation. For any two pacing profiles $\bm{a},\bm{a}'\in\Omega$, define $\bm{a}''=(\max(a_1,a_1'),\cdots,\max(a_n,a_n'))$, we prove that $\bm{a}''\in\Omega$. Specifically, for each $i\in[n]$, we have either $a''_i=a_i$ or $a''_i=a_i'$. If $a''_i=a_i$ holds, we have $f_i(a_{-i})\geq a_i=a''_i$. Observe that $a''_{-i}$ is point-wise weakly larger than $a_{-i}$, by the monotonicity of $f_i$, we have $f_i(a''_{-i})\geq f_i(a_{-i})\geq a''_i$. If $a''_i=a_i'$ holds, we similarly have $f_i(a''_{-i})\geq f_i(a'_{-i})\geq a'_i=a''_i$. Therefore we have $f_i(a''_{-i})\geq a''_i$ for all $i\in[n]$, i.e., $\bm{a}''\in\Omega$.

    Next, let $\bm{a}^*$ be the point-wise maximum of $\Omega$, i.e., $a^*_i=\max_{a\in\Omega}a_i$ for each $i\in[n]$. We know that $\bm{a}^*\in\Omega$, and that $\bm{a}^*$ is point-wise weakly larger than any pacing equilibrium, so it remains to prove that $\bm{a}^*$ is a pacing equilibrium, or equivalently, $\vec{f}(\bm{a}^*)=\bm{a}^*$. 
    Suppose for the contradiction that $\vec{f}(\bm{a}^*)\neq\bm{a}^*$. We denote $\bm{a}'=\vec{f}(\bm{a}^*)$. Since $\bm{a}^*\in\Omega$, for any $i\in[n]$, we have $a'_i=f_i(a^*_{-i})\geq a^*_i$. It follows that $a'_{-i}$ is point-wise weakly larger than $a^*_{-i}$, so by monotonicity we have $f_i(a'_{-i})\geq f_i(a^*_{-i})=a'_i$. Then we have $\bm{a}'\in\Omega$, contradicting with that $\bm{a}^*$ is the point-wise maximum of $\Omega$. So we have $f(\bm{a}^*)=\bm{a}^*$, and thus $\bm{a}^*$ is a pacing equilibrium. 

%The proof of Property 3 is put in appendix.
%move the proof of property 3 to appendix
\textbf{Property 3:}
Let $\bm{a}^*$ denote the point-wise maximum pacing equilibrium. We prove that for any pacing equilibrium $\bm{a}\in[0,1]^n$, $\bm{a}$ and $\bm{a}^*$ induce the same outcome in all requests.

Define $\delta_i=a^*_i-a_i$ for each $i\in[n]$.
We use the following lemma, which will be proved at the end of this proof.
\begin{lemma}
\label{claim:1}
Define $\Gamma(\bm{a};\bm{a}^*;q)=\sum_{i\in[n]}\delta_i\cdot\frac1{\tau_i}v_i(q)(x_i(\bm{a}^*;q)-x_i(\bm{a};q))$.
When $q\sim D$, it holds almost surely that $\Gamma(\bm{a};\bm{a}^*;q)\geq 0$. Moreover, define events $E_1=\{\exists i\in[n], x_i(\bm{a}^*;q)\neq x_i(\bm{a};q)\}$, and $E_2=\{\Gamma(\bm{a};\bm{a}^*;q)>0\}$, then $E_1$ implies $E_2$, i.e., $\Pr_{q\sim D}[E_1\cap \neg E_2]=0$.
\end{lemma}
To prove property 3, suppose for contradiction that $\Pr_{q\sim D}[\exists i\in[n], x_i(\bm{a}^*;q)\neq x_i(\bm{a};q)]>0$, then by \Cref{claim:1}, we have
% both $\Pr_{q\sim D}[\sum_{i\in[n]}\delta_i\cdot\frac1{\tau_i}v_i(q)(x_i(\bm{a}^*;q)-x_i(\bm{a};q))\geq 0]=1$ and $\Pr_{q\sim D}[\sum_{i\in[n]}\delta_i\cdot\frac1{\tau_i}v_i(q)(x_i(\bm{a}^*;q)-x_i(\bm{a};q))>0]>0$.,
\begin{align*}
\Pr_{q\sim D}[\sum_{i\in[n]}\delta_i\cdot\frac1{\tau_i}v_i(q)(x_i(\bm{a}^*;q)-x_i(\bm{a};q))\geq 0]=1,\\
\Pr_{q\sim D}[\sum_{i\in[n]}\delta_i\cdot\frac1{\tau_i}v_i(q)(x_i(\bm{a}^*;q)-x_i(\bm{a};q))>0]>0.
\end{align*}
This leads to $\E_{q\sim D}[\sum_{i\in[n]}\delta_i\cdot\frac1{\tau_i}v_i(q)(x_i(\bm{a}^*;q)-x_i(\bm{a};q))]>0$. Thus, there exists some $i\in[n]$ such that $$\E_{q\sim D}[\delta_i\cdot\frac1{\tau_i}v_i(q)(x_i(\bm{a}^*;q)-x_i(\bm{a};q))]>0.$$
It follows that $\delta_i>0$ and $\E_{q\sim D}[\frac1{\tau_i}v_i(q)(x_i(\bm{a}^*;q)-x_i(\bm{a};q))]>0$. The latter inequality implies that $\E_{q\sim D}[p_i(\bm{a}^*;q)]>\E_{q\sim D}[p_i(\bm{a};q)]$, while $\delta_i>0$ implies that $0\leq a_i<a^*_i\leq 1$. 

However, since $a_i<1$ and $a_i^*>0$, we have $\E_{q\sim D}[p_i(\bm{a};q)]\geq \frac{B_i}{T}$ and $\E_{q\sim D}[p_i(\bm{a}^*;q)]\leq \frac{B_i}{T}$ by the definition of pacing equilibrium. Then we have $\E_{q\sim D}[p_i(\bm{a};q)]\geq\E_{q\sim D}[p_i(\bm{a}^*;q)]$, which contradicts. 

Therefore, when $q\sim D$, with probability $1$ we have $x_i(\bm{a}^*;q)=x_i(\bm{a};q)$. This also immediately implies $p_i(\bm{a}^*;q)=p_i(\bm{a};q)$.

Now it remains to prove \Cref{claim:1}. Denote the rank scores under a pacing profile by $s_i(\bm{a};q)=a_i\frac1{\tau_i}v_i(q)+\kappa\cdot e_i(q)$.
When $q\sim D$, by \Cref{assumption:nomass}, we can assume without loss of generality that there are no ties under $\bm{a}$, i.e., all rank scores $s_1(\bm{a};q),\cdots,s_n(\bm{a};q)$ are distinct. 

To prove the lemma, we firstly observe that for any two sellers $i_1,i_2\in[n]$, if $s_{i_1}(\bm{a};q)<s_{i_2}(\bm{a};q)$ while $s_{i_1}(\bm{a}^*;q)>s_{i_2}(\bm{a}^*;q)$, then $\delta_{i_1}\frac1{\tau_{i_1}}v_{i_1}(q)=s_{i_1}(\bm{a}^*;q)-s_{i_1}(\bm{a};q)>s_{i_2}(\bm{a}^*;q)-s_{i_2}(\bm{a};q)=\delta_{i_2}\frac1{\tau_{i_2}}v_{i_2}(q)$. Informally speaking, this implies that when the ranking of two sellers are reversed, the corresponding swap of their allocated slots always contributes a positive term to $\Gamma(\bm{a};\bm{a}^*;q)$.

Formally, let $\phi$ and $\phi^*$ each denote the permutation of the $n$ sellers in descending order of the rank scores, under $\bm{a}$ and $\bm{a}^*$, respectively. Specifically, we have $x_{i}(\bm{a};q)=\sum_{j\in[m]}\beta_j\I[\phi(j)=i]$ and $x_{i}(\bm{a}^*;q)=\sum_{j\in[m]}\beta_j\I[\phi^*(j)=i]$, respectively.

Let $\psi$ denote the permutation that transforms $\phi$ into $\phi^*$, i.e., $\phi^*=\psi\circ\phi$. $\psi$ can be decomposed into a sequence of adjacent swaps. Specifically, consider the swaps that occur when applying a bubble sort algorithm on $\phi$ to sort it in descending order of $s_{i}(\bm{a}^*;q)$. This decomposes $\psi$ into a sequence of swaps of adjacent positions $\left((j^{(k)},j^{(k)}+1)\right)_{k=1,\cdots,K}$, where $K\leq\frac{n(n-1)}{2}$ is the total number of swaps.

Let $\phi^{(k)}$ denote the permutation after the $k$-th swap, with $\phi^{(0)}=\phi$ and $\phi^{(K)}=\phi^*$. 
When making every swap $(j^{(k)},j^{(k)}+1)$, let $i_1^{(k)}=\phi^{(k)}(j^{(k)})$ and $i_2^{(k)}=\phi^{(k)}(j^{(k)}+1)$ represent the two sellers whose positions are swapped. 
Then for each $k\in[K]$, we have the conditions $s_{i_1^{(k)}}(\bm{a};q)<s_{i_2^{(k)}}(\bm{a};q)$ and  $s_{i_1^{(k)}}(\bm{a}^*;q)>s_{i_2^{(k)}}(\bm{a}^*;q)$ of the earlier observation, and thus it holds that $\delta_{i_1^{(k)}}\frac1{\tau_{i_1^{(k)}}}v_{i_1^{(k)}}(q)>\delta_{i_2^{(k)}}\frac1{\tau_{i_2^{(k)}}}v_{i_2^{(k)}}(q)$. Now, observe that $\Gamma(\bm{a};\bm{a}^*;q)$ can be written as
\begin{align*}
&\Gamma(\bm{a};\bm{a}^*;q)
=\sum_{k=1}^{K}\left(\beta_{j^{(k)}}-\beta_{j^{(k)}+1}\right)\left(\frac{\delta_{i_1^{(k)}}v_{i_1^{(k)}}(q)}{\tau_{i_1^{(k)}}}-\frac{\delta_{i_2^{(k)}}v_{i_2^{(k)}}(q)}{\tau_{i_2^{(k)}}}\right).
\end{align*}
Here we define $\beta_j=0$ for all $j>m$. Since $\beta_{j^{(k)}}-\beta_{j^{(k)}+1}\geq 0$ and $\frac{\delta_{i_1^{(k)}}v_{i_1^{(k)}}(q)}{\tau_{i_1^{(k)}}}>\frac{\delta_{i_2^{(k)}}v_{i_2^{(k)}}(q)}{\tau_{i_2^{(k)}}}$ hold for all $k\in[K]$, we have $\Gamma(\bm{a};\bm{a}^*;q)\geq 0$.

Moreover, when $E_1$ happens, there must exist some $k_0\in [K]$ that $j^{(k_0)}\leq m$ and $\beta_{j^{(k_0)}}-\beta_{j^{(k_0)}+1}>0$, because otherwise all sellers' obtained exposure rates will remain unchanged. It follows that $\Gamma(\bm{a};\bm{a}^*;q)>0$, i.e., $E_2$ is implied.

This completes the proof.
\end{proof}

\subsection{Properties under Pacing Equilibrium}
With the basic properties of pacing equilibrium, we show the key property of LPA that its allocation outcome under the pacing equilibrium is exactly the optimal allocation maximizing the platform objective.
To state this theorem, we first define some notations:
Define $\mathcal{F}_{A}=\{(x_{i,j})_{i\in[n],j\in[m]}\in[0,1]^{n\times m}:
\forall i\in[n],\sum_{j\in[m]}x_{i,j}\leq 1;~
\forall j\in[m], \sum_{i\in[n]}x_{i,j}\leq 1\}$ as the feasible space of allocations, and define $\mathcal{F}_\beta=\{(x_i)_{i\in[n]}\in\bbR^n:\exists (x_{i,j})_{i\in[n],j\in[m]}\in \mathcal{F}_{A}, x_i=\sum_{j\in[m]}\beta_jx_{i,j}\}$ as the feasible space of allocated exposure rate for the sellers.
%\looseness=-1
\begin{theorem}
\label{thm:platform-objective-optimal}
The allocation induced by any pacing equilibrium is an optimal solution to the following platform objective maximization problem:
\begin{align*}
&\max_{(x_{i,q})_{i\in[n],q\in Q}}\sum_{i\in[n]}\left(\min\left\{\bbE_{q\sim D}\left[\frac1{\tau_i}v_i(q)x_{i,q} \right],\frac{B_i}{T}\right\}+\kappa \bbE_{q\sim D}\left [e_i(q)x_{i,q} \right] \right )\\
&\text{s.t. }\forall q,~(x_{i,q})_{i\in[n]}\in \mathcal{F}_{\beta}.
\end{align*}
\end{theorem}

\begin{proof}
The platform objective maximization problem can be written as the following program with a convex feasible region and a linear objective function:
\begin{align}
&\max_{(x_{i,q})_{i\in[n],q\in Q},(y_i)_{i\in[n]}} \sum_{i\in[n]}(y_i+\kappa \bbE_{q\sim D}[e_i(q)x_{i,q}]) \label{eq:optproblem}\\
\text{s.t. }&(x_{i,q})_{i\in[n]}\in \mathcal{F}_{\beta},\ \forall q\in Q;\nonumber\\
&y_i\leq \frac1{\tau_i}\bbE_{q\sim D}[v_i(q)x_{i,q}],\ \forall i\in[n];\nonumber\\
&y_i\leq \frac{B_i}{T},\ \forall i\in[n].\nonumber
\end{align}

We write the Lagrangian form relaxing the constraints about $y_i$:
\begin{align*}
    L(x,y,\lambda,\mu)
    =\sum_{i\in[n]}&\left(\left(1-\lambda_i-\mu_i\right)y_i+\kappa \bbE_{q\sim D}[e_i(q)x_{i,q}]+
    \lambda_i\frac1{\tau_i}\bbE_{q\sim D}[v_i(q)x_{i,q}]+\mu_i\frac{B_i}{T}\right).
\end{align*}

By strong duality, there exists optimal dual variables $\lambda^*,\mu^*$, such that the solution $(x^*,y^*)=\argmax_{x,y}L(x,y,\lambda^*,\mu^*)$ is also the optimal solution of the original problem (\ref{eq:optproblem}). Moreover, by KKT conditions, we have that 
\begin{align*}
    0=\frac{\partial L(x,y,\lambda^*,\mu^*)}{\partial y_i}=1-\lambda_i^*-\mu_i^*,~\forall i\in[n].
\end{align*}
This implies that $\mu_i^*=1-\lambda_i^*$, so we can rewrite
\begin{align*}
L(x,y,\lambda^*,\mu^*)=&\sum_{i\in[n]}\bbE_{q\sim D}\left[\left(\lambda_i^*\frac1{\tau_i}v_i(q)+\kappa e_i(q)\right)x_{i,q}\right]
+\sum_{i\in[n]}(1-\lambda_i^*)\frac{B_i}{T}.
\end{align*}

Observe that since $x^*$ maximizes $L(x,y,\lambda^*,\mu^*)$, it coincides with the allocation that sort the sellers in the descending order $\lambda_i^*\frac1{\tau_i}v_i(q)+\kappa e_i(q)$, i.e., the allocation of LPA when the bids are given by $b_i(q)=\lambda^*\frac1{\tau_i}v_i(q)$, which happens exactly when $\lambda^*$ is viewed as the pacing profile.
Formally, we have $x_i(\lambda^*;q)=x^*_{i,q}$ and $p_i(\lambda^*;q)=\frac1{\tau_i}v_i(q)x^*_{i,q}$.

Next, we prove that, if we view $\lambda^*$ as a pacing profile, it forms a pacing equilibrium. We check the conditions in \Cref{def:pacing equilibrium}. 

(1) If $\lambda^*_i>0$, then by KKT conditions, both $y_i^*=\bbE_{q\sim D}[\frac1{\tau_i}v_i(q)x^*_{i,q}]$ and $y_i^*\leq \frac{B_i}{T}$ hold. It follows that $\bbE_{q\sim D}[\frac1{\tau_i}v_i(q)x^*_{i,q}]\leq \frac{B_i}{T}$.

(2) If $\lambda^*_i<1$, then $\mu^*_i>0$. By KKT conditions, both $y_i^*=\frac{B_i}{T}$ and $y_i^*\leq \bbE_{q\sim D}[\frac1{\tau_i}v_i(q)x^*_{i,q}]$ hold, and then it follows that $\bbE_{q\sim D}[\frac1{\tau_i}v_i(q)x^*_{i,q}]\geq \frac{B_i}{T}$.

Recall that $p_i(\lambda^*;q)=\frac1{\tau_i}v_i(q)x^*_{i,q}$, then by \Cref{def:pacing equilibrium}, $\lambda^*$ is a pacing equilibrium.

By property 3 in \Cref{thm:pacing equilibrium exist-unique}, all pacing equilibria induce the same allocation and payment outcomes in all requests. Particularly, the induced allocation will be $(x^*_{i,q})_{i\in[n],q\in Q}$, i.e., the optimal solution maximizing the platform objective.
\end{proof}

In the following theorem, we show that the joint system of LPA and uniform bidding is incentive compatible for stock-constrained value maximizing sellers, incentivizing them to truthfully report both budget and ROI constraints, under the reasonable assumption that the pacing equilibrium is achieved.

\begin{theorem}\label{thm:IC-for-stock-constrained}
Under the stock-constrained value maximizer model, for each seller $i$, truthfully reporting both budget $B_i$ and ROI $\tau_i$ is always optimal for the seller's target under the pacing equilibrium.
\end{theorem}
\begin{proof}
For any seller $i\in[n]$, fix the budget and ROI reported by other sellers, let $\bm{a}$ denote the pacing equilibrium when seller $i$ truthfully reports $B_i,\tau_i$, and let $\bm{a}'$ denote the pacing equilibrium when seller $i$ reports any other $B_i',\tau_i'$. We show that seller $i$'s target is weakly better under $\bm{a}$ than under $\bm{a}'$. We discuss two cases:

(a) $\bbE_{q\sim D}[p_i(\bm{a};q)]<B_i/T$. In this case, it must hold that $a_i=1$, which leads to the bidding strategy of $b_i(q)=\frac1{\tau_i}v_i(q)$. Suppose for contradiction that $\bm{a}'$ strictly improves seller $i$'s target, then the induced bidding strategy $b_i'(q)=a_i'\frac1{\tau_i'}v_i(q)$ must satisfy that $\smash{a_i'\frac1{\tau_i'}>\frac1{\tau_i}}$. As $a_i'\leq 1$, we have $\smash{\frac1{\tau_i'}>\frac1{\tau_i}}$. However, by the payment rule of LPA, the payment when reporting $\tau_i'$ is $\bbE_{q\sim D}[p_i'(\bm{a}';q)]=\frac1{\tau_i'}\bbE_{q\sim D}[x_i(\bm{a};q)v_i(q)]>\frac1{\tau_i}\bbE_{q\sim D}[x_i(\bm{a};q)v_i(q)]$, so this outcome violates seller $i$'s ROI constraint, which contradicts.

(b) $\bbE_{q\sim D}[p_i(\bm{a};q)]\geq B_i/T$. In this case, by LPA's payment rule, we have $\bbE_{q\sim D}[x_i(\bm{a};q)v_i(q)]\geq\tau_iB_i/T$, i.e., the obtained sales reaches the stock limit, which is already the best possible outcome for stock-constrained seller $i$.

This completes the proof.
\end{proof}

\subsection{Computability of Pacing Equilibrium}
Generally, the existence of a pacing equilibrium does not guarantee that the auto-bidding system will converge to it. However, we show that the pacing equilibrium in LPA is computationally tractable given the distribution, which informally suggests that the auto-bidding system is likely to converge to the pacing equilibrium.
% Generally speaking, the existence of the pacing equilibrium does not imply the convergence of the auto-bidding system to such equilibrium. 
% For example, the pacing equilibrium in second-price auction 
% Here we demonstrate that the pacing equilibrium under LPA is computationally tractable, as an indirect evidence that the auto-bidding system will likely to converge to the pacing equilibrium in practice.
\begin{definition}[$\epsilon$-pacing equilibrium]
    A pacing profile $\bm{a}\in[0,1]^n$ is called an $\epsilon$-pacing equilibrium, if for each $i\in[n]$, $a_i$ satisfies that:
\begin{enumerate}
\item If $a_i>\epsilon$, then $\E_{q\sim D}\left[p_{i}\left(a_i-\epsilon,a_{-i};q\right)\right]\leq \frac{B_i}{T} $.
\item If $a_i<1-\epsilon$, then $\E_{q\sim D}\left[p_{i}\left(a_i+\epsilon,a_{-i};q\right)\right]\geq \frac{B_i}{T}$.
\end{enumerate}
% \begin{enumerate}
% \item If $a_i>0$, then $\E_{q\sim D}\left[p_{i}\left(b_1(q),\cdots,b_n(q);q\right)\right]\leq \frac{B_i}{T} + \epsilon$.
% \item If $a_i<1$, then $\E_{q\sim D}\left[p_{i}\left(b_1(q),\cdots,b_n(q);q\right)\right]\geq \frac{B_i}{T} - \epsilon$.
% \end{enumerate}
\end{definition}

\begin{algorithm}[t]
\SetAlgoLined
\caption{Computing Pacing Equilibrium}

\label{alg:simplified}

Initialize $\bm{a} \leftarrow 1^n$\;

 \While{$\exists i\in[n],~a_i>\epsilon\land \E_{q\sim D}\left[p_{i}\left(a_i-\epsilon,a_{-i};q\right)\right]>\frac{B_i}{T}$}{
    Update $a_i\gets \min\{a_i'\geq 0:\E_{q\sim D}\left[p_{i}\left(a_i',a_{-i};q\right)\right]\geq\frac{B_i}{T}\}$\;
 }
 Output $\bm{a}$\;
\end{algorithm}

\begin{theorem}
\label{thm:simple alg convergence}
When the request distribution $D$ is known, \Cref{alg:simplified} computes an $\epsilon$-pacing equilibrium in at most $n\epsilon^{-1}$ iterations.
\end{theorem}
\begin{proof}
Observe that each iteration will decrease some $a_i$ for at least $\epsilon$, with $\sum_{i}a_i=n$ initially and $\sum_ia_i\geq 0$ at the end, so the algorithm terminates in at most $\frac{n}{\epsilon}$ iterations.

By the monotonicity of payment, when updating some $a_i$, $p_i(\bm{a};q)$ is non-increasing, while $p_j(\bm{a};q)$ is non-decreasing for all $j\neq i$. 
Therefore, it holds at any moment of the algorithm that, for all $i\in[n]$, either $a_i=1$, or $\E_{q\sim D}\left[p_{i}\left(\bm{a};q\right)\right]\geq \frac{B_i}{T}$
So when the algorithm ends, $a_i<1$ implies $\E_{q\sim D}\left[p_{i}\left(\bm{a};q\right)\right]\geq \frac{B_i}{T}$. Moreover, by the ending condition, $a_i>\epsilon$ implies $\E_{q\sim D}\left[p_{i}\left(a_i-\epsilon,a_{-i};q\right)\right]\leq \frac{B_i}{T}$. That is, the output $\bm{a}$ is an $\epsilon$-pacing equilibrium.
\end{proof}

\section{Experiment Results}
In this section, we evaluate the performance of LPA under the QuanZhanTui scenario in comparison with conventional auctions. All auctions share the same allocation rule as LPA, and differ only in the payment rules.
We answer the following questions: %\emph{Does LPA outperform conventional auctions in the QuanZhanTui scenario?}
\emph{How does it perform in improving Pareto frontiers?} and \emph{How close is its performance to the theoretical optimum in an auto-bidding system?}
\subsection{Simulation Implementation}
To simulate the outcome of the auction process under different auctions, we consider both online and offline settings. In the online setting, we simulate the repeated auction process over the sequence of requests, where each seller's bids are controlled by a pacing algorithm. We initialize the pacing factors with $\smash{a_i^{(0)}}=1$, and update them using the following rule:
$$
a_i^{(t)}=\max\left(0,\min \left (1,a_i^{(t-1)}+\eta\left (\frac{B_i}{T}-p_i^{(t)}\right)\right)\right),
$$
where $\eta$ is the learning rate, set to $0.01$. This rule aims to control the average spending rate to $B_i/T$.
We observe that this update rule results in stable adjustments of pacing factors and effective budget control.
%\looseness=1
In the offline setting, we assume access to all the requests, and directly compute the pacing equilibrium. We evaluate the auction outcomes under the pacing equilibrium, where all bids are set according to the equilibrium pacing factors.
%\looseness=1

\subsection{Datasets}
We simulate the auction process on both synthetic and real-world datasets, each consisting of the sellers' budget and ROI constraints, the sales quantity $v_i(q_t)$ and the user experience $e_i(q_t)$ in the requests, as well as other environment parameters.
\looseness=1

We generate the synthetic dataset with $n=10$ sellers and $T=10^5$ user requests, where each request has $m=3$ slots. The exposure rates are set as $\smash{(\beta_1,\beta_2,\beta_3)=(1,0.6,0.4)}$. 
% For each seller $i\in[n]$, we generate the budget constraints $B_i$ independently following the log-normal distribution, such that $\log(B_i)\sim\mathcal{N}(0,1)$. 
For each seller $i\in[n]$, we generate the budget constraints $B_i$ independently from $U[0,10^5]$.
ROI constraints $\tau_i$ are generated so that $\frac{1}{\tau_i}$ independently follows the distribution $U[0,1]$.
We set $\kappa=1$ and generate the expected sales quantity $v_i(q_t)$ and user experience $e_i(q_t)$ as follows. We first generate a product feature vector $\smash{C^{p}_{i}\in\bbR^{d_C}}$ for each seller $i\in[n]$ and a user feature vector $\smash{C^{u}_{t}\in\bbR^{d_C}}$ for each request $q_t$, where $d_C$ is set to $10$. Each $\smash{C^{p}_{i}}$ and $\smash{C^{u}_{t}}$ follows the standard multivariate normal distribution, i.e., each component is independently drawn from the distribution $\mathcal{N}(0,1)$. Then we set $\smash{v_i(q_t)=\exp(1/(\sqrt{d_C})\langle C^{p}_{i},C^{u}_{t}\rangle)}$, where $\smash{\langle C^{p}_{i},C^{u}_{t}\rangle}$ is the inner product of $\smash{C^{p}_{i}}$ and $\smash{C^{u}_{t}}$. We finally draw independently from $U[0,1]$ two parameters $r^{i,t}_1$ and $r^{i,t}_2$ and set $e_i(q_t)=\frac12(r^{i,t}_1v_i(q_t)+ 0.2r^{i,t}_2)$. The expectations of $v_i(q_t)$, $e_i(q_t)$ are $0.8$, $0.4$ respectively.
\looseness=1

% \subsection{Real-world Dataset}
We synthesize a real-world dataset from the QuanZhanTui service on an e-commerce platform. $n=50$ sellers and $T=10^5$ requests are sampled from the real-world data distribution. Specifically, we sample budget and ROI constraints from the real data. The budgets are then scaled to maintain the average budget per request. For each request $q_t$, we generate $v_i(q_t)$ and $e_i(q_t)$ by sampling the predicted GMV and CTR from the real auctions. We set $\kappa=0.5$. Each request has $m=6$ slots with exposure rates $(\beta_1,\cdots,\beta_6)=(1,0.8,0.7,0.6,0.5,0.4)$.
%\looseness=1
% Using this synthesized dataset, we evaluate the performance of LPA, GFP, GSP and VCG in both online and offline setting, the results are listed in \Cref{tab:real}.

\begin{table*}[ht]
    \centering
    \begin{tabular}{|c|cccc|}
\hline\textbf{Auction}&$\LW$&$\kappa\cdot\UE$&$\OBJ$&$\REV$\\
\hline LPA (offline)& 315608 (100.0\%)& 221100 (100.0\%)& 536709 (100.0\%)& 315608  (100.0\%)\\
\hline LPA & 313901 (99.5\%)& 221428 (100.1\%)& 535329 (99.7\%)& 313901  (99.5\%)\\
\hline GFP & 296417 (93.9\%)& 219863 (99.4\%)& 516280 (96.2\%)& 292670  (92.7\%)\\
\hline GSP & 284041 (90.0\%)& 216508 (97.9\%)& 500550 (93.3\%)& 167262  (53.0\%)\\
\hline VCG & 281969 (89.3\%)& 216156 (97.8\%)& 498125 (92.8\%)& 119574 (37.9\%)\\
\hline
    \end{tabular}
    \begin{tabular}{|c|cccc|}
\hline\textbf{Auction}&$\LW$&$\kappa\cdot\UE$&$\OBJ$&$\REV$\\
\hline LPA (offline)& 76913 (100.0\%)& 33084 (100.0\%)& 109998 (100.0\%)& 76913 (100.0\%)\\
\hline LPA& 76699 (99.7\%)& 33055 (99.9\%)& 109754 (99.8\%)& 76699 (99.7\%)\\
\hline GFP& 73667 (95.8\%)& 31791 (96.1\%)& 105459 (95.9\%)& 73106 (95.1\%)\\
\hline GSP& 71368 (92.8\%)& 30380 (91.8\%)& 101748 (92.5\%)& 45557 (59.2\%)\\
\hline VCG& 69128 (89.9\%)& 29053 (87.8\%)& 98181 (89.3\%)& 32865 (42.7\%)\\
\hline
    \end{tabular}
    \caption{Experiment results under synthetic (top) and real-world (bottom) data}
    \label{tab:tabs-synth-and-real}
\end{table*}

\begin{figure}[ht]
    \centering
    \sffamily\fontfamily{phv}\selectfont
    
    % 第一幅图（上）
    \begin{tikzpicture}
        \begin{axis}[
            width=0.49\linewidth, 
            height=0.4\linewidth, % 高度略减以适应两图
            xlabel={$\LW$ (\(\times 10^5 \))},
            ylabel={$\UE$ (\( \times 10^5 \))},
            legend pos=south west,
            legend style={font=\scriptsize},
            tick label style={font=\sffamily\scriptsize},
            label style={font=\sffamily\small},
            grid=both,
            every axis plot/.append style={thick},
        ]
        \addplot[mark=*, color=blue] coordinates{(3.2726023462181155, 1.833232552469008) (3.2683725333935825, 1.898999276961077) (3.2515634452567213, 1.9988586146247667) (3.2116642482965063, 2.114744915233896) (3.1429622693866723, 2.213671958002078) (3.0362065471543755, 2.2906810234915977) (2.910905762285014, 2.3359780800413326) (2.8054428445257953, 2.3554332127131357) (2.7231750880945533, 2.362987093459028)};
        \addplot[mark=square*, color=red] coordinates{(3.0035462442391827, 1.9389270958431035) (3.0048477026492573, 1.971082025375876) (3.003472958232883, 2.0264966559878816) (2.9954980664043354, 2.105696116977009) (2.964887360373329, 2.198281465204162) (2.9074873608279326, 2.277599089039408) (2.8359810070892526, 2.328440480531243) (2.7666324574718795, 2.353074058613263) (2.699976667570361, 2.362473190446565)};
        \addplot[mark=diamond*, color=olive] coordinates{(2.8652883544711165, 1.9181407987265064) (2.8674221351135025, 1.9492233021402874) (2.866849011274124, 1.9978815322317713) (2.858611457450816, 2.072638452802528) (2.8390966765559047, 2.1644744049970153) (2.807697610161592, 2.251656680145055) (2.7723790674481337, 2.3149229782710905) (2.723722323204054, 2.3477456407838617) (2.6741206974811167, 2.360759511867456)};
        \addplot[mark=triangle*, color=orange] coordinates{(2.8294276211148874, 1.9139746193761518) (2.8310769729184035, 1.9437925073179727) (2.8307328421359665, 1.9928523197693422) (2.829175632740959, 2.067758590550262) (2.8204249702113064, 2.1605377838614666) (2.7976265150728192, 2.250885282332995) (2.7675691310475523, 2.3143744535471686) (2.7216872297977335, 2.3477104099636747) (2.6731351468426245, 2.3607503490143564)};
        \legend{LPA, GFP, GSP, VCG}
        \end{axis}
    \end{tikzpicture}
    \quad
    % \vspace{0.2cm} % 两图间距
    % 第二幅图（下）
    \begin{tikzpicture}
        \begin{axis}[
            width=0.49\linewidth, 
            height=0.4\linewidth, % 高度略减以适应两图
            xlabel={$\LW$ (\(\times 10^5 \))},
            ylabel={$\UE$ (\( \times 10^5 \))},
            legend pos=south west,
            legend style={font=\scriptsize},
            tick label style={font=\sffamily\scriptsize},
            label style={font=\sffamily\small},
            grid=both,
            every axis plot/.append style={thick},
        ]
        \addplot[mark=*, color=blue] coordinates{(8.0642581604252, 4.611992992152919) (8.00146951231115, 5.321230001115215) (7.877449525519027, 6.01729506147668) (7.6701239250734945, 6.610480420628331) (7.36393589058886, 7.049700658071262) (6.967901855915976, 7.33381657597156) (6.531112530823492, 7.49253421805182) (6.104190397640927, 7.57088381280889) (5.73876487373856, 7.604678676782458)};
        \addplot[mark=square*, color=red] coordinates{(7.600318683194279, 4.644609801838695) (7.571150376442126, 5.116329417887612) (7.503075083760006, 5.729110136711334) (7.3669497008946925, 6.35680389089446) (7.133753589663068, 6.879778254033365) (6.813551756961686, 7.2382523777457335) (6.4314431601546485, 7.446703894319453) (6.047949625896153, 7.551194173396159) (5.710161569231444, 7.597081371036568)};
        \addplot[mark=diamond*, color=olive] coordinates{(7.331972883147591, 4.534673831989631) (7.308996654049025, 4.9131934147569325) (7.253520914132302, 5.452652141476923) (7.135459115339078, 6.076446668515019) (6.937991908130025, 6.645803495503979) (6.653229284101126, 7.074204590230013) (6.309741665643145, 7.346224892214869) (5.964658730807578, 7.495877451529348) (5.655162880156644, 7.5699950554459186)};
        \addplot[mark=triangle*, color=orange] coordinates{(7.065474594397086, 4.436981577573862) (7.046666161615656, 4.744063819042774) (7.002832579544309, 5.215046442279669) (6.914228324677376, 5.8126341621404265) (6.7537017750697546, 6.4161007445872045) (6.516087583573859, 6.91359078486924) (6.220275332500375, 7.252780425835937) (5.912120221326528, 7.449102398233849) (5.629305253040645, 7.551022206928958)};
        \legend{LPA, GFP, GSP, VCG}
        \end{axis}
    \end{tikzpicture}
    
    \caption{Pareto frontier analysis with synthetic (left) and real-world (right) datasets.}
    \label{fig:combined-vertical}
\end{figure}
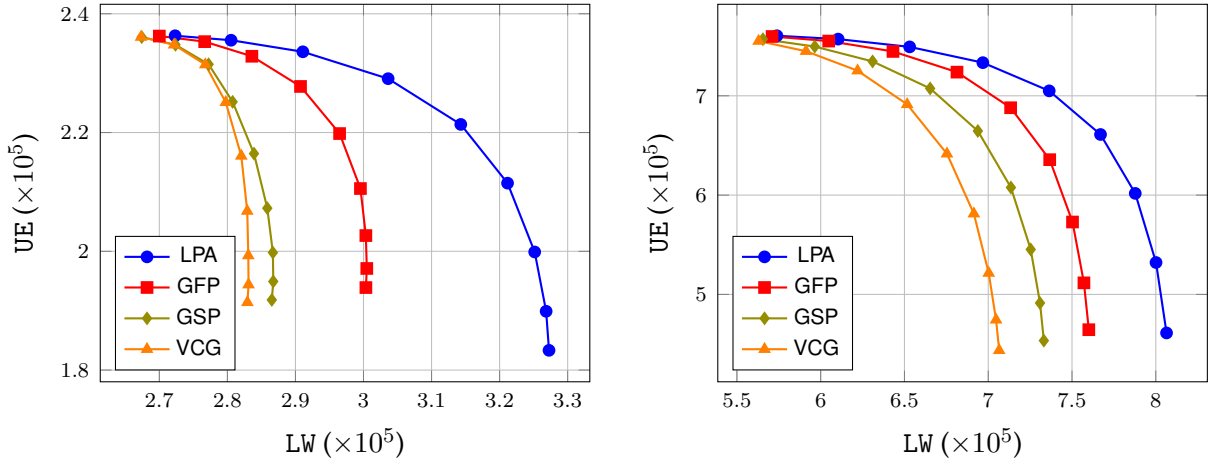
\subsection{Results}
%We evaluate LPA on both synthetic and real-world datasets, considering both online and offline settings.
We compare LPA with three baseline auctions: VCG \cite{RePEc:bla:jfinan:v:16:y:1961:i:1:p:8-37vcg}, GSP \cite{DBLP:conf/www/WilkensCN17} and GFP \cite{DBLP:journals/mor/DuttingFP19}. The evaluations focus on four key metrics: platform objective ($\OBJ$), liquid welfare ($\LW$), user experience ($\UE$), and revenue ($\REV$). We generate 10 independent versions of each dataset and report the average values of each metric.
%Furthermore, to demonstrate that LPA improves the Pareto frontier, we conduct experiments with $\kappa$ varying over a range of values: $2^{-4}, 2^{-3}, \cdots, 2^{4}$.

\textbf{Performance Evaluation.}
Results are presented in \Cref{tab:tabs-synth-and-real}. For the online setting, we present the performance of all auctions. For the offline setting, we report the performance of LPA, which represents the theoretically optimal platform objectives. We compare the results of each auction with those achieved by LPA under pacing equilibrium. Since we observe that the offline results are consistently similar to the online ones (with $<1\%$ difference), we omit the results of baseline auctions in the offline setting. We observe that LPA consistently outperforms all baseline auctions in all four metrics. Notably, LPA simultaneously improves both $\LW$ and $\UE$ compared to baseline auctions.
In addition, LPA gains a significant improvement in $\REV$. This is favorable to e-commerce platforms that seek to optimize platform profitability. Moreover, the performance of LPA in the online setting is close to the theoretically optimal in the offline setting, with at most a $0.5\%$ discrepancy in all metrics. This indicates that LPA is highly compatible with auto-bidding systems, generating near-optimal outcomes even when paired with simple auto-bidding algorithms.

\textbf{Pareto Frontier Analysis.}
We also examine LPA's improvement of the Pareto frontier. Specifically, we vary $\kappa$ over the range: $2^{-4}, 2^{-3}, \cdots, 2^{4}$ and simulate the auction process in the online setting. The results are presented in \Cref{fig:combined-vertical}. For each auction, we plot a line where each point corresponds to $\LW$ and $\UE$ achieved for a particular $\kappa$. These lines visualize the trade-off between the two metrics when $\kappa$ changes, with larger $\kappa$ generally leading to larger $\UE$ and smaller $\LW$. We observe that when $\kappa$ is large, all auctions achieve a similar $\UE$, while LPA has a slightly better $\LW$. As $\kappa$ decreases, the differences in $\LW$ become more apparent. The results on both the synthetic and real-world datasets show that LPA consistently improves the Pareto frontier.

\section{Conclusion}
In this paper, we study the mechanism design problem for the platform-wide marketing service QuanZhanTui. We propose the stock-constrained seller model and the Liquid Payment Auction (LPA). Through theoretical analysis, we establish key properties of the pacing equilibrium induced by LPA in a uniform bidding environment, such as existence and computability. Furthermore, we demonstrate that LPA optimally allocates platform-wide traffic under this equilibrium and incentivizes sellers to truthfully report their constraints. Our experimental evaluation results further validate the superior performance of LPA across multiple aspects including liquid welfare, user experience, and platform revenue, as well as its compatibility to practical pacing-based auto-bidding systems. We note that this work not only advances the theoretical understanding of the auction design in QuanZhanTui, but also offers a practical solution for industrial implementations. We hope that this research contributes valuable information to both academia and industry.

\section*{Acknowledgements}
This work is supported by the National Natural Science Foundation of China (Grant No. 62172012), supported by Alibaba Group through the Alibaba Innovative Research Program, and supported by Peking University-Alimama Joint Laboratory of AI Innovation. 
The authors thank Zhaohua Chen for helpful discussions and all anonymous reviewers for their helpful feedback.
\newpage

\bibliography{ref}
\bibliographystyle{plain}

% \appendix
% \input{body/appendix}
\end{document}